\documentclass[american,a4paper,twoside,12pt]{article}

\usepackage{times,color,babel}
\usepackage{amssymb,amsthm,amsmath,bm}
\usepackage[T1]{fontenc}
\usepackage[flushmargin]{footmisc}   
\usepackage{dsfont}    
\usepackage{graphicx}
\usepackage{enumerate}
\usepackage[colorlinks=true,hyperindex=true]{hyperref}
\usepackage{fancyhdr}
\usepackage{a4wide}
\newcommand{\matrice}[2][cccccccccccccccccc]{%
\left(\hspace{-1mm}\begin{array}{#1}#2\\ \end{array}\hspace{-1mm}\right)}

\let\rho=\varrho

\def\ie{{\sl i.e.,} }

\def\t{{\rm t}}

\def\e{{\rm e}}

\def\d{{\rm d}}
\def\p{{\rm p}}
\def\t{{\rm t}}

\def\real{{\mathbb R}}
\def\eref#1{(\ref{#1})}

\def\Theorem{{Theorem}}

\newtheorem{theorem}{Theorem}

\newtheorem{lemma}[theorem]{Lemma}

\begin{document}
\title{\bf Temperature and voltage probes\\ far from equilibrium}

\author{\sc Ph.~A.~Jacquet$^{1}$ and C.-A. Pillet$^{2}$
\\ \\ \\
$^1$Department of Physics, Kwansei Gakuin University, Sanda 669-1337 Japan
\\ \\
$^2$Aix-Marseille Univ, CPT, 13288 Marseille cedex 9, France\\
CNRS, UMR 7332, 13288 Marseille cedex 9, France\\
Univ Sud Toulon Var, CPT, B.P. 20132, 83957 La Garde cedex, France\\
FRUMAM
}
\def\today{}
\maketitle
\thispagestyle{empty}
\vskip 1.5truecm
\begin{quote}
{\noindent\bf Abstract.}
We consider an open system of non-interacting electrons consisting of a small sample connected 
to several reservoirs and temperature or voltage probes. We study the non-linear system of equations
that determines the probe parameters. We show that it has a unique solution, which can be 
computed with a fast converging iterative algorithm. We illustrate our method with two well-known 
models: the three-terminal system and the open Aharovov-Bohm interferometer.
\end{quote}

\section{Introduction}
\label{Introduction}

Thermodynamic quantities such as entropy, temperature, or chemical potential play a fundamental
role in our understanding of equilibrium phenomena. They are given sound microscopic meanings 
within the framework of equilibrium statistical mechanics. The concept of local thermal equilibrium
allows us in principle to define these quantities in interacting systems close to equilibrium.
However, extending these definitions far from equilibrium and/or to non-interacting systems 
where local equilibrium does not make sense is a much more delicate issue (see the discussions 
in \cite{Jou,Ru1,Ru2}). In this paper, we shall consider an operational point of view, giving
to local intensive parameters the values measured by external probes.
 
Such an experimental approach is well-known in the mesoscopic community: in the description of 
electric transport in a multi-terminal system, in which all the terminals have the same temperature 
(typically $T=0$), one often introduces a \emph{voltage probe} \cite{ButtikerFourT} to sense the local 
electrochemical potential by connecting an additional electronic reservoir under the 
\emph{zero electric current condition}: the chemical potential of the probe is 
tuned so that there is no net average electric current into it. In the same spirit, setting all the 
terminals to the same chemical potential, a \emph{temperature probe} is obtained by requiring that 
the temperature of the corresponding reservoir is tuned such that there is no average heat current 
into it \cite{AE}. 

In the scattering approach of Landauer and B\"uttiker (see Section 2), the existence and uniqueness of such 
parameters are usually accepted on physical grounds, but we think it is important and interesting 
to obtain a rigorous mathematical foundation for these fundamental parameters. In the linear 
response regime, a rigorous proof has recently been given \cite{Thesis}. Here, we shall extend 
these results to the far-from-equilibrium regime and furthermore provide an efficient numerical 
method for computing their values. More explicitly, this paper is organized as follows: in Section~2,
we describe the framework, in Sections~3 and 4, we present our main results and their proofs. Finally, 
in Section 5, we illustrate our method by considering two well-known models: the three-terminal 
system and the open Aharovov-Bohm interferometer.

\noindent{\bf Acknowledgments.}
This work was partially supported by the Swiss National Foundation, the Japan Society for the 
Promotion of Science and the Agence Nationale de la Recherche (contract HAMMARK 
ANR-09-BLAN-0098-01). C.-A. P. wishes to thank Tooru Taniguchi for hospitality at Kwansei 
Gakuin University where part of this work was done.

\section{Framework}
\label{Model}

We consider a multi-terminal mesoscopic system, that is, a small system $\cal S$ connected through 
leads to several infinitely extended particle reservoirs (see Figure \ref{Fig1}). We assume that the 
transport properties of this system can be described within the Landauer-B\"uttiker framework. 
More precisely, we consider $N$ reservoirs in equilibrium at inverse temperature 
$\beta_i=1/(k_{\rm B}T_i)$ and chemical potential $\mu_i$ ($i=1,\ldots,N$). The corresponding
Fermi-Dirac distributions are
\begin{equation}
f_i(E)=f(E,\beta_i,\mu_i) = \left[1+\e^{\beta_i(E-\mu_i)}\right]^{-1}.
\end{equation} 
For simplicity, here and in what follows we set the Boltzmann constant $k_{\rm B}$, the Planck 
constant $h$ and the elementary charge $e$ to unity.

In the Landauer-B\"uttiker formalism, one neglects all interactions among the particles and
considers the small system as a scatterer for the particles emitted by the reservoirs.
Thus, the small system is completely characterized by the one-particle on-shell scattering matrix 
$S(E)=\left[S_{ij;mn}(E)\right]$, where the indices $i,j\in\{1,\dots,N\}$ label the outgoing/incoming 
terminals and for each pair $(i,j)$ the indices $m \in \{1, \dots, M_i(E)\}$ and $n \in \{1, \dots, M_j(E)\}$
label the open channels in terminals $i$ and $j$, respectively. The matrix element $S_{ij;mn}(E)$ is 
the probability \emph{amplitude} for a particle with energy $E$ incident in channel $n$ of terminal 
$j$ to be transmitted into channel $m$ of terminal $i$. The corresponding \emph{total} transmission 
probability $t_{ij}(E)$ that a particle with energy $E$ goes from terminal $j$ to terminal $i$ is given 
by \cite{Buttiker-Symmetry}
\begin{equation}\label{transmission prob}
t_{ij}(E) = \sum_{m=1}^{M_i(E)}
\sum_{n=1}^{M_j(E)} |S_{ij;mn}(E)|^2.
\end{equation}
The unitarity of the scattering matrix immediately yields the following identities,
\begin{equation}\label{Property Transmission}
\sum_{i=1}^{N} t_{ij}(E) = M_j(E), \hspace{5mm} \sum_{j=1}^{N} t_{ij}(E)=M_i(E).
\end{equation}

\begin{figure}[htbp]
\begin{center}
\includegraphics[width=7cm]{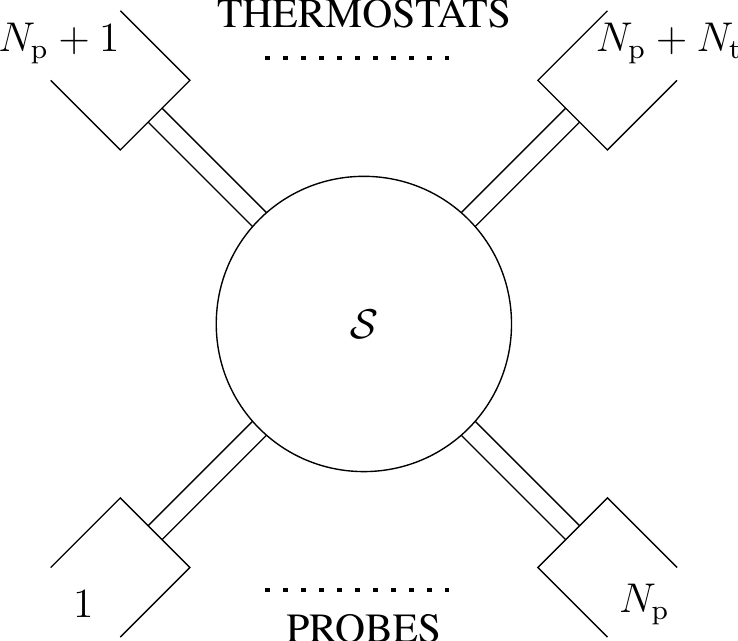}
\caption{A multi-terminal system: The sample is connected through leads to $N=N_\p+N_\t$
particle reservoirs. The reservoirs $1, \dots, N_\p$ are probes and the reservoirs 
$N_\p+1, \dots, N_\p+N_\t$ are thermostats driving the system out of equilibrium.}\label{Fig1}
\end{center}
\end{figure}

The expected stationary electric and heat currents in lead $i \in \{1,\dots,N\}$ are
given by the celebrated Landauer-B\"uttiker formulas \cite{ButtikerScattering,Butcher},
\begin{eqnarray}
I_i &=&  \sum_{j=1}^{N} \int\!\!\left[t_{ji}(E) f_i(E) - t_{ij}(E) f_j(E)\right] \d E,\label{Eq I}\\
J_i &=& \sum_{j=1}^{N} \int\!\!\left[t_{ji}(E) f_i(E) - t_{ij}(E) f_j(E)\right] (E-\mu_i) \d E.\label{Eq J}
\end{eqnarray}
From a phenomenological point of view, these expressions can be easily 
understood: $t_{ji}(E) f_i(E)$ is the average number of particles with energy $E$ that are 
transmitted from terminal $i$ to terminal $j$, and $t_{ij}(E) f_j(E)$ is the same but from terminal 
$j$ to terminal $i$. Therefore, $I_i$ ($J_i$) is the \emph{net} average electric (heat) current in
lead $i$, counted positively from the $i$-th terminal to the system. 
Mathematical derivations of these formulas (including existence of a stationary regime)
rest on the assumption that the leads are infinitely extended and act as reservoirs
\cite{CA,N}.

Considerable interest has been devoted to electric transport in which all the terminals have the same 
temperature. In this context, an important concept emerged: the \emph{voltage probe} 
\cite{ButtikerFourT}. A voltage probe is a large physical component used in mesoscopic experiments 
to sense the local electrochemical potential. Theoretically, such a probe is modeled as a reservoir 
under the \emph{zero electric current condition}: the chemical potential of the probe is 
tuned so that there is no net average electric current into it. If all the terminals have the same 
temperature, then in general there will also be a heat current into the probe. In this 
case, we will  consider this heat current as dissipation. 
In the same spirit, setting all the terminals to the same chemical potential, a 
\emph{temperature probe} is obtained by requiring that the temperature of the corresponding 
reservoir is tuned such that there is no average heat current into it. Note that in this case, there 
may be some charge dissipation into the temperature probe.

Let us decompose the $N$ terminals as follows: the first $N_\p$ reservoirs are temperature 
or voltage probes and the remaining $N_t$ reservoirs are the thermostats maintaining the 
system out of equilibrium (see Figure \ref{Fig1}). In the \emph{voltage probe configuration}, all the 
reservoirs are at the same inverse temperature $\beta=\beta_1=\cdots=\beta_N$, the chemical 
potentials of the thermostats $\vec\mu_\t=(\mu_{N_\p+1},\ldots,\mu_{N_\p+N_\t})$ are given 
and we have to determine the probe parameters $\vec\mu_\p=(\mu_1,\ldots,\mu_{N_\p})$ such 
that $\vec I_\p=(I_1,\ldots,I_{N_\p})=\vec 0$.
Similarly, in the \emph{temperature probe configuration}, all reservoirs are at the same chemical
potential $\mu=\mu_1=\cdots=\mu_N$, the thermostat inverse temperatures 
$\vec\beta_\t=(\beta_{N_\p+1},\ldots,\beta_{N_\p+N_\t})$ are fixed and we have to determine the 
probe parameters $\vec\beta_\p=(\beta_1,\ldots,\beta_{N_\p})$ in order to satisfy 
$\vec J_\p=(J_1,\ldots,J_{N_\p})=\vec0$.

To our knowledge, no result is available on these two problems beyond the linear approximation
around global equilibrium (\ie linear response theory, see \cite{AE}). The same remark applies
to other approaches to the determination of local intensive thermodynamic parameters (see {\sl e.g.,}
\cite{BILP,SI}).

\section{Results}
\label{Results}

Note that in both configurations the \emph{self-consistency condition}
$$
\vec I_\p(\beta,\vec\mu_\t;\vec\mu_\p)=\vec0,
\qquad\text{or}\qquad
\vec J_\p(\mu,\vec\beta_\t;\vec\beta_\p)=\vec0,
$$
is a system of $N_\p$ non-linear equations with $N_\p$ unknown. From a mathematical perspective,
it is not at all obvious 
that such a system admits a solution. Moreover, if a solution exists, it may not be unique.
Our main result ensures existence and uniqueness of reasonable solutions to these
equations.

We shall make the following general assumptions on the lead Hamiltonians and scattering
matrix:
\begin{enumerate}[{\bf (A)}]
\item There exists a constant $E_0$ such that $M_j(E)=0$ for all $E\le E_0$ and $j\in\{1,\dots, N\}$.
\item $M_j(E)\leq C (1+|E|)^\eta$ for some constants $C$ and $\eta$ and all
$E$ and $j\in\{1,\dots, N\}$.
\item For every $j\in\{1,\dots,N_\p\}$ there exists a set $\mathcal{E}_j\subset\real$ of positive Lebesgue 
measure such that
$$
\widetilde M_j(E)=\sum_{i=N_\p+1}^{N_\p+N_\t} t_{ij}(E) > 0,
$$
for all $E\in\mathcal{E}_j$.
\end{enumerate}

\noindent
Condition (A) merely asserts that the lead Hamiltonians are bounded below. Condition
(B) limit the growth of the number of open scattering channels as function of the energy
and is satisfied by any physically reasonable lead Hamiltonian. Finally, Condition (C) can
be roughly rephrased as follows: any probe is connected through an open scattering channel 
to some thermostat.

To formulate our main result, let us denote
\begin{align*}
\underline{\mu}&=\min\{\mu_{N_\p+1},\ldots,\mu_{N_\p+N_\t}\},\\ 
\overline{\mu}&=\max\{\mu_{N_\p+1},\ldots,\mu_{N_\p+N_\t}\},
\end{align*}
the minimal/maximal chemical potential of the thermostats and define in the same way
$\underline{\beta}$ and $\overline{\beta}$.

\begin{theorem}\label{MainThm} Under the above assumptions, the following hold:
\begin{enumerate}[{\rm (1)}]
\item The self-consistency condition $\vec I_\p(\beta,\vec\mu_\t;\vec\mu_\p)=\vec0$ has a unique 
solution $\vec\mu_\p=\vec\mu_\p(\beta,\vec\mu_\t)$ in the set
$\{(\mu_1,\ldots,\mu_{N_\p})\,|\,\mu_j\in[\underline{\mu},\overline{\mu}]\}$.
\item The self-consistency condition $\vec J_\p(\mu,\vec\beta_\t;\vec\beta_\p)=\vec0$ has a unique
solution $\vec\beta_\p=\vec\beta_\p(\mu,\vec\beta_\t)$ in the set
$\{(\beta_1,\ldots,\beta_{N_\p})\,|\,\beta_j\in[\underline{\beta},\overline{\beta}]\}$.
\item In both cases the solution can be computed by means of a rapidly convergent
algorithm (see the next sections for details).
\end{enumerate}
\end{theorem}

{\noindent\bf Remarks.} 1. The restriction on the solution is physically reasonable. We do not expect
a temperature probe to measure a value below the smallest thermostat temperature or above 
the highest one. The same remark applies to voltage probes.

\medskip\noindent
2. An alternative approach to probing local intensive parameters is to adjust \emph{both} 
$\vec\beta_\p$ and $\vec\mu_\p$ in such a way that the electric \emph{and} heat
currents vanish: $\vec I_\p=\vec J_\p=0$. Such probes thus measure simultaneously the temperature 
\emph{and} the chemical potential. Note that in this case there is no dissipation at all into the probes.
Our method does not apply directly to this situation, basically because the function 
$f(E,\beta,\mu)-f(E,\beta',\mu')$ does not preserve its sign as $E$ varies if $\beta\not=\beta'$
\emph{and} $\mu\not=\mu'$. To our knowledge, no result is available for such dual probes
beyond the linear approximation around global equilibrium (see \cite{Jacquet,Thesis}).


\section{Proofs}
\label{Proofs}

Let us discuss first the voltage probe configuration. Using the relations \eref{Property Transmission}, 
the self-consistency condition may be written as
\begin{equation}\label{Self-consistency}
I_i = \sum_{j=1}^{N} \int f(E,\beta,\mu_j) [M_j(E) \delta_{ij}-t_{ij}(E)] \d E=0,
\end{equation} 
for $i=1,\dots,N_\p$. Under Assumptions (A) and (B),
\begin{equation}
\mu\mapsto X_j(\mu) = \int f(E,\beta,\mu) M_j(E) \d E,
\end{equation}
defines a strictly increasing continuous function. We shall 
denote by $X\mapsto\mu_j(X)$ the reciprocal function. The key idea of our approach 
is to work with the variable $\vec X=(X_1(\mu_1),\ldots,X_{N_\p}(\mu_{N_\p}))$ instead of 
$\vec\mu_\p$. 

Let $\vec F: \real^{N_\p}\rightarrow\real^{N_\p}$ be defined as
\begin{align*}
F_i(\vec{X})&=\sum_{j=1}^{N_\p}\int f(E,\beta,\mu_j(X_j)) t_{ij}(E) \d E\\
&+\sum_{j=N_\p+1}^{N_\p+N_\t}\int f(E,\beta,\mu_j) t_{ij}(E) \d E.
\end{align*} 
Then, we can rewrite the self-consistency condition \eref{Self-consistency} as a fixed point equation
\begin{equation}
\vec F(\vec{X})=\vec{X}.
\label{Fpoint}
\end{equation}
Set $\underline{X}_j=X_j(\underline{\mu})$, $\overline{X}_j=X_j(\overline{\mu})$ and denote
$$
\Sigma=\{\vec{X}=(X_1,\ldots,X_{N_\p})\,|\,X_j\in[\underline{X}_j,\overline{X}_j]\}.
$$
Notice that the condition $\vec X\in\Sigma$ is equivalent to $\mu_j\in[\underline{\mu},\overline{\mu}]$
for all $j\in\{1,\ldots,N_\p \}$.
\begin{lemma}\label{Map}
$\vec F(\Sigma) \subset \Sigma$. 
\end{lemma}

\begin{proof}
Let $\vec{X}\in\Sigma$. The monotony of $\mu\mapsto f(E,\beta,\mu)$ implies
$f(E,\beta,\mu_j(X_j))\leq f(E,\beta,\overline{\mu})$ 
for $j=1,\ldots,N_\p$ and $f(E,\beta,\mu_j)\leq f(E,\beta,\overline{\mu})$ for 
$j=N_\p+1,\ldots,N_\p+N_\t$. The identities \eref{Property Transmission} yields
$$
F_i(\vec{X})\leq\int f(E,\beta,\overline{\mu}) M_i(E)  \d E = \overline{X}_i.
$$
Proceeding similarly, one shows
$$
F_i(\vec{X})\geq \int f(E,\beta,\underline{\mu}) M_i(E) \d E = \underline{X}_i.
$$
\end{proof}

Under Conditions (A) and (B), the function $\vec F$ is continuous. Since $\Sigma$ is compact 
and convex, it follows from Lemma \ref{Map} and the Brouwer fixed point theorem that 
\eref{Fpoint} has a solutions $\vec X^\star\in\Sigma$. We shall use Condition (C) to ensure
uniqueness of this solution. In the next lemma, we use the norm $\|\vec X\|=\sum_{j=1}^{N_\p}|X_j|$.

\begin{lemma}\label{ReallyContract}
Under Assumptions (A), (B) and (C) there exists a constant $\theta<1$ such that
\begin{equation}
\|\vec F(\vec{X})-\vec F(\vec{X}')\|\le\theta\|\vec{X}-\vec{X}'\|,
\label{ContractBound}
\end{equation}
for any $\vec{X},\vec{X}'\in\Sigma$.
\end{lemma}
\begin{proof}
Denote by $D(\vec X)$ the derivative of the map $\vec F$ at $\vec X$. Then one has
$$
\vec F(\vec{X})-\vec F(\vec{X}')=\int_0^1 D(t\vec X+(1-t)\vec X')(\vec X-\vec X')\d t.
$$
Since $\Sigma$ is convex, the estimate \eref{ContractBound} holds for any 
$\vec{X},\vec{X}'\in\Sigma$ with
$$
\theta=\max_{\vec X\in\Sigma}\|D(\vec X)\|,
$$
where the matrix norm is given by
$$
\|D(\vec X)\|=\max_{1\le j\le N_\p}\sum_{i=1}^{N_\p}
\left|D_{ij}(\vec X)\right|.
$$
A simple calculation yields
$$
D_{ij}(\vec X)=\frac{\partial F_i}{\partial X_j}(\vec X)
=\frac{\int g(E,\beta,\mu_j(X_j))t_{ij}(E)\d E}{\int g(E,\beta,\mu_j(X_j))M_j(E)\d E},
$$
where the function $g(E,\beta,\mu)=\partial_\mu f(E,\beta,\mu)$ is strictly positive. It follows that
$$
\sum_{i=1}^{N_\p}\left|D_{ij}(\vec X)\right|
\le1-\frac{\int\widetilde M_j(E)g(E,\beta,\mu_j(X_j))\d E}{\int M_j(E)g(E,\beta,\mu_j(X_j))\d E},
$$
and hence
$$
\theta\le1-\min_{\mu\in[\underline{\mu},\overline{\mu}]\atop 1\le j\le N_\p}
\frac{\int\widetilde M_j(E)g(E,\beta,\mu)\d E}{\int M_j(E)g(E,\beta,\mu)\d E}.
$$
Condition (C) clearly implies that $\theta<1$.
\end{proof}

It follows from Lemmas \ref{Map}, \ref{ReallyContract} and the Banach fixed point theorem
that  \eref{Fpoint} has a unique solution $\vec X^\star$ in $\Sigma$. Moreover, the sequence
of iterates $\vec X_n=\vec F(\vec X_{n-1})$ converges to $\vec X^\star$ for any initial
value $\vec X_0\in\Sigma$ with the estimate
$$
\|\vec X_n-\vec X^\star\|\le\theta^n\|\vec X_0-\vec X^\star\|.
$$

In the temperature probe configuration, one may proceed in a completely similar way in terms 
of the functions
\begin{equation}
\beta\mapsto Y_j(\beta)=\int f(E,\beta,\mu) (E-\mu) M_j(E)\d E,
\end{equation}
their reciprocal $Y\mapsto\beta_j(Y)$ and
\begin{align*}
G_i(\vec{Y})=&\sum_{j=1}^{N_\p}\int f(E,\beta_j(Y_j),\mu) (E-\mu)t_{ij}(E)\d E\\
&+\sum_{j=N_\p+1}^{N_\p+N_\t} \int f(E,\beta_j,\mu) (E-\mu)t_{ij}(E)\d E.
\end{align*} 
A natural set $\Sigma$ can then be defined as before. The crucial observation is that the function
$\partial_\beta f(E,\beta,\mu) (E-\mu)$ has a constant sign.

\medskip\noindent{\bf Remark.} Strictly speaking, Lemma \ref{ReallyContract} does not hold
at zero temperature because the Fermi function $f$ is not positive in this case. Nevertheless,
one easily shows that, under Assumptions (A)--(C), the estimate
$$
\|\vec F(\vec X)-\vec F(\vec X')\|<\|\vec X-\vec X'\|,
$$
holds for $\vec X,\vec X'\in\Sigma$, $\vec X\not=\vec X'$ provided 
$[\underline{\mu},\overline{\mu}]\subset\cap_j\mathcal{E}_j$. The uniqueness of the fixed point 
$\vec X^\star$ immediately follows. Moreover, it also follows that the sequence of iterates 
$\vec X_n=\vec F(\vec X_{n-1})$ converges to $\vec X^\star$ for any choice of $\vec X_0\in\Sigma$,
although without {\sl a priori} control on the speed of convergence.


\section{Examples}

As a first example, let us consider the one-channel three-terminal system represented in Figure \ref{Fig2},
where two thermostats (2 and 3) drive the system (a perfect lead) out of equilibrium and a probe (1) 
is connected to the system by a $3 \times 3$ scattering matrix $S$.

\begin{figure}[htbp]
\begin{center}
\includegraphics[width=12cm]{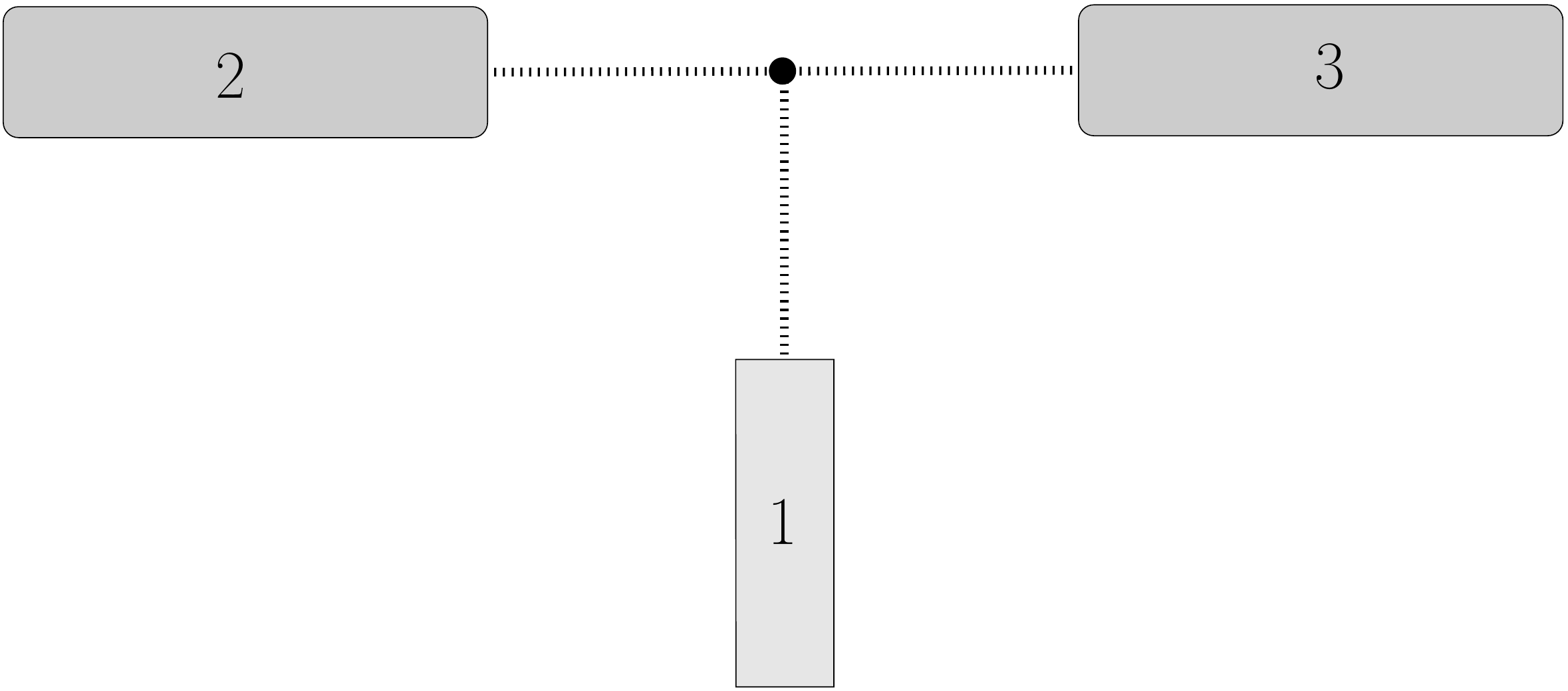}
\caption{A one-channel system with two thermostats (2 and 3) and one probe (1).}\label{Fig2}
\end{center}
\end{figure}

Let us consider the energy-independent scattering matrix introduced in \cite{Buttiker-Small}:
\begin{equation}\label{Scattering matrix Markus}
S = \matrice{
-(a+b)  & \sqrt{\varepsilon} & \sqrt{\varepsilon}\\
\sqrt{\varepsilon} & a & b \\ \sqrt{\varepsilon} & b & a}~,
\end{equation}
where $a=\frac{1}{2} (\sqrt{1-2\varepsilon}-1)$, $b=\frac{1}{2} (\sqrt{1-2\varepsilon}+1)$ and
$\varepsilon \in (0,\frac{1}{2}]$. Here, $\varepsilon = 0$ corresponds to the uncoupled situation 
(which is excluded) and $\varepsilon = \frac{1}{2}$ to the maximally coupled one. Let us set 
$T = T_1 = T_2 = T_3$ and define the energy interval in \eref{Eq I}--\eref{Eq J} as $[0,\infty)$. 
If $T=0$, then in the linear regime one can compute analytically the self-consistent parameter 
$\mu^*_1$ \cite{Buttiker-Symmetry,Jacquet}:
\begin{equation}\label{Linear mu2}
\mu_1^*(T=0,\mu_2,\mu_3,\varepsilon) = \frac{\mu_2 + \mu_3}{2} + \mathcal{O}(|\mu_2-\mu_3|^2)~.
\end{equation}
We have checked that our numerical results are consistent with the relation \eref{Linear mu2}. 
In the non-linear regime, we made the following observations: Let $T>0$, $\mu_2, \mu_3 \in \real$ 
be fixed, then the sequence $\{F^n(X_{0})\}_{n=0}^{\infty}$, with $X_{0} \in \Sigma$, converges 
and gives rise to a value $\mu_1^*$ independent of $\varepsilon$ and conveniently written as
\begin{equation}
\mu_1^*(T,\mu_2,\mu_3,\varepsilon) = \frac{\mu_2 + \mu_3}{2}  + \mathcal{N}(T,\mu_2,\mu_3)~,
\end{equation}
where the function $\mathcal{N}(T,\mu_2,\mu_3)$ measures the non-linearity. Note, in particular,
that the weak coupling limit $\varepsilon \rightarrow 0$ does not lead to a different value of $\mu_1^*$.
Since, by \Theorem~1,
$$
\mu_1^*(T,\mu_2,\mu_3,\varepsilon)\in[\min\{\mu_2,\mu_3\},\max\{\mu_2,\mu_3\}],
$$
one deduces that $\mathcal{N}(T,\mu_2,\mu_3) \in [-\Delta\mu/2,\Delta\mu/2]$, with 
$\Delta \mu = |\mu_2-\mu_3|$. In Figures~\ref{Fig3} and \ref{Fig4}, we have plotted the temperature and 
potential dependence of $\mathcal{N}(T,\mu_2,\mu_3)$, respectively. Let us recall that 
$\mathcal{N}(T,\mu_2,\mu_3)=0$ corresponds to the linear case.

\begin{figure}[htbp]
\begin{center}
\includegraphics[width=12cm]{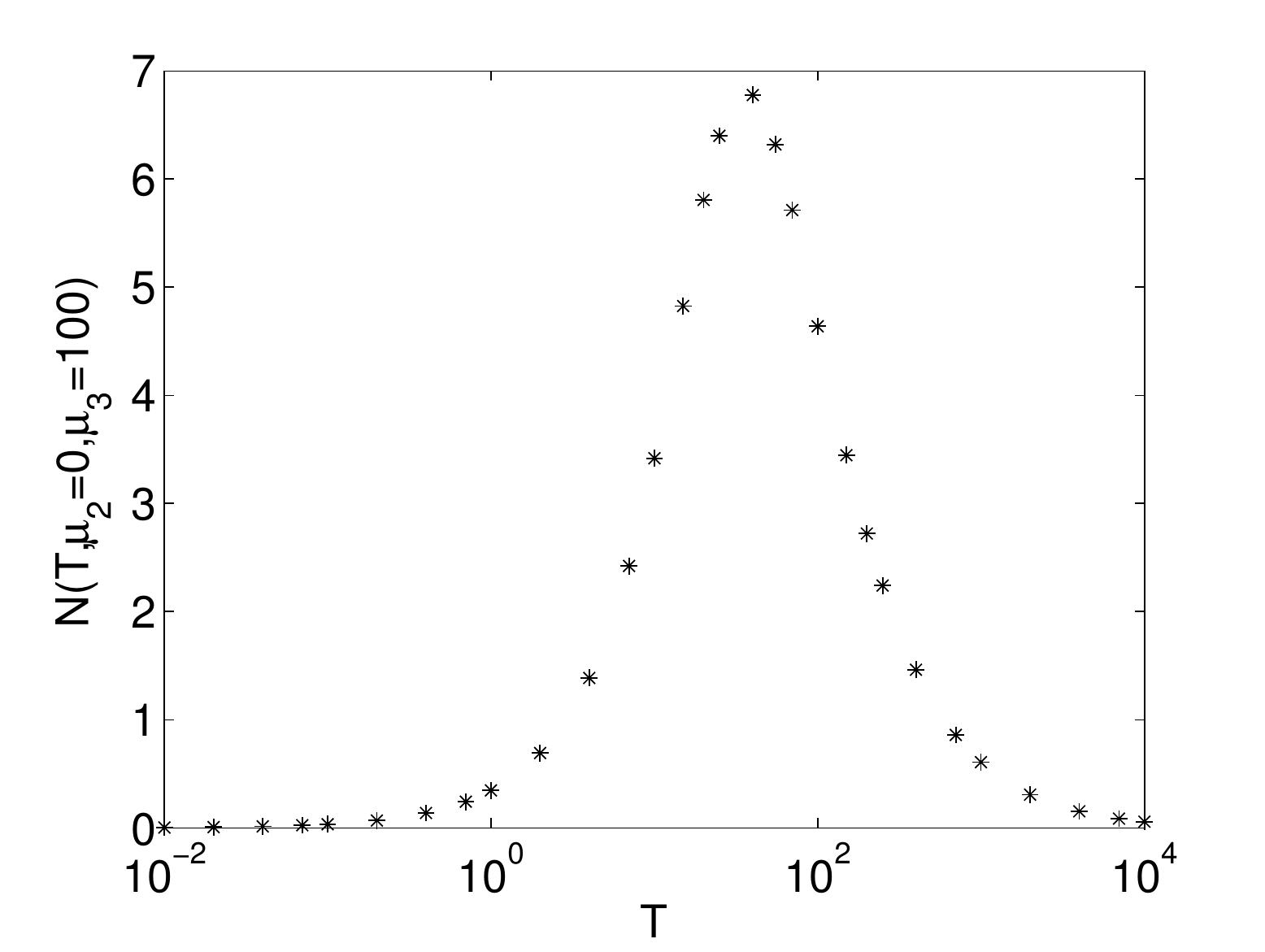}
\caption{The temperature dependence of $\mathcal{N}(T,\mu_2=0,\mu_3=100)$.}\label{Fig3}
\end{center}
\end{figure}

\begin{figure}[htbp]
\begin{center}
\includegraphics[width=12cm]{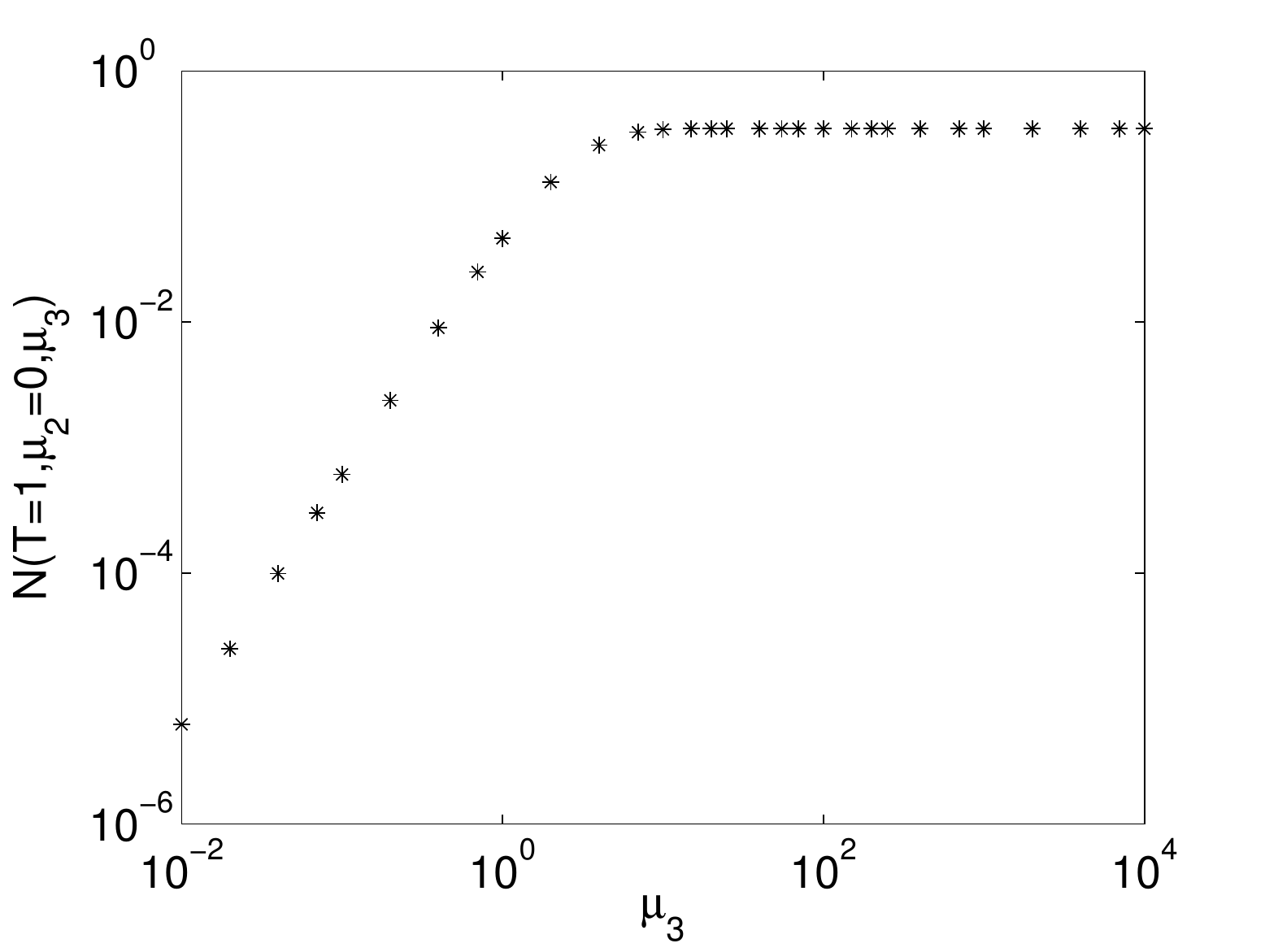}
\caption{The potential dependence of $\mathcal{N}(T=1,\mu_2=0,\mu_3)$.}\label{Fig4}
\end{center}
\end{figure}

Note that the curve in Figure~\ref{Fig4} reaches a constant value $\mathcal{N}_\infty(T)=\lim_{\mu_3 \rightarrow \infty} \mathcal{N}(T,\mu_2=0,\mu_3)$ as $\mu_3$ increases. Interestingly, this is also the case for other values of $T$ and we observed the following scaling law:
\begin{equation}
\mathcal{N}_\infty(\lambda T)=\lambda \mathcal{N}_\infty(T)~, \hspace{5mm} \forall\lambda > 0.
\end{equation}

If we attach more probes to the lead and describe all the connection points in terms of the 
\emph{same} scattering matrix $S$ (see \cite{Jacquet} for the construction of the global scattering 
matrix), then we found that all the probes measure the \emph{same} value, as if all the probes were 
connected to the same point, but in general this value does not coincide with the one-probe 
measurement (since adding more probes somehow perturbs the system). This phenomenon can be 
easily understood: for example, if two probes are attached to the lead, then one can compute 
analytically the global $4 \times 4$ transmission matrix $\{t_{ij}(E)\}$, and one finds that it is 
symmetric and that $t_{31}(E)=t_{32}(E)=t_{41}(E)=t_{42}(E)$. This means that the two probes 
are equally coupled to the left and right thermostats and, consequently, that $\mu_1^*=\mu_2^*$. Note, 
however, that this is not true in general.

As a second example, we consider an Aharonov-Bohm (AB) ring threaded by a magnetic flux $\phi$
and with a quantum dot (QD) embedded in one of its arms. This system has been subject to intensive
investigations both in the independent electron approximation \cite{BILP,TT1} and including interaction 
effects \cite{KAKI,TT2,Ho}. We shall study a discrete (tight binding) independent electron model 
closely related to the work by Aharony {\sl et al.}\;\cite{Aharony} (see Figure~\ref{Fig5}). This theoretical model 
is supposed to imitate an experimental setup \cite{Schuster}. It is assumed that a gate voltage $V$ 
is applied on the QD, allowing to vary the energies of its eigenstates.
\begin{figure}[htbp]
\begin{center}
\includegraphics[scale=0.3]{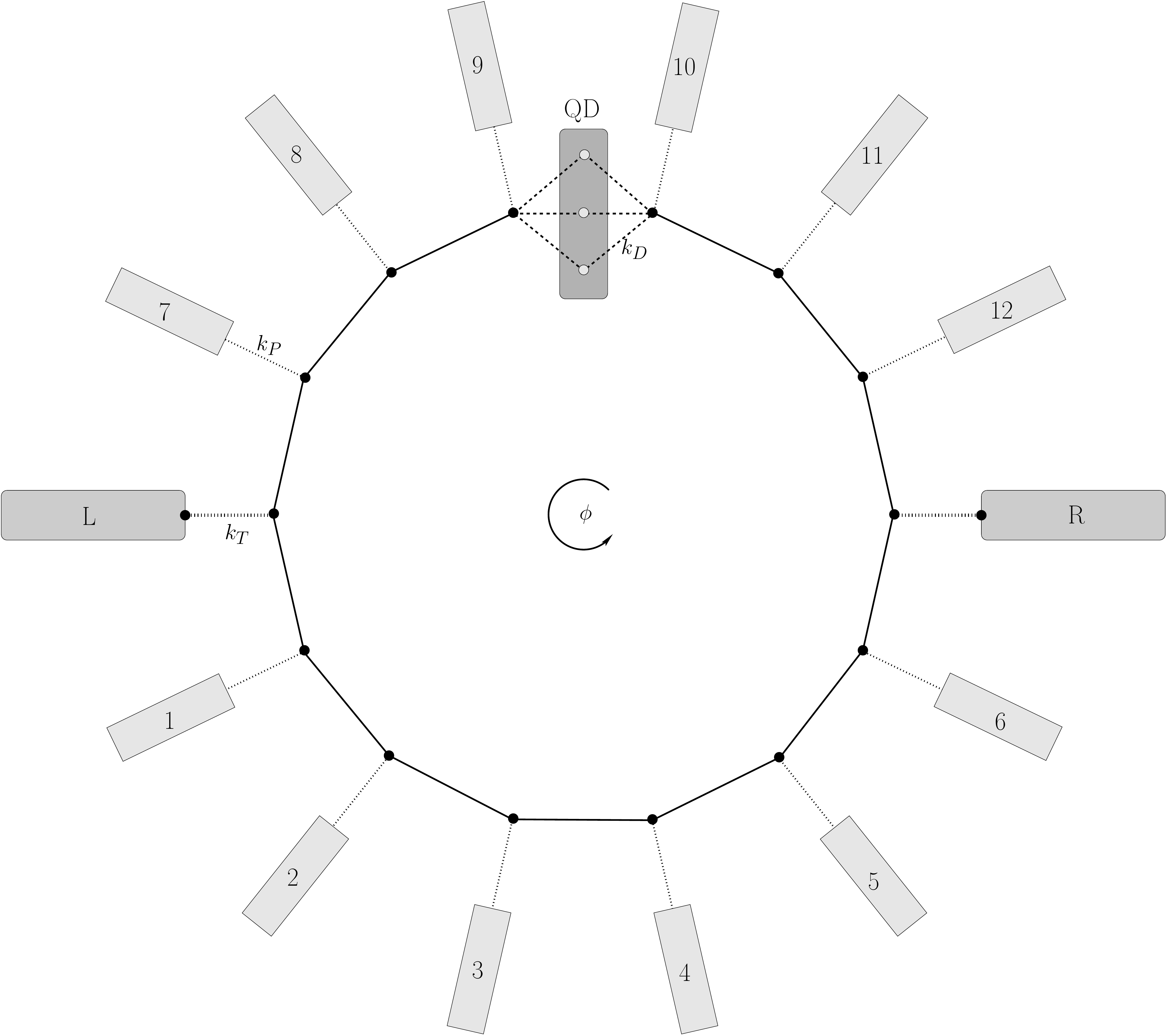}
\caption{The open AB interferometer. A magnetic flux $\phi$ crosses the ring and a QD is
placed in the upper branch of the ring. The terminals 1 to 12 are voltage probes and the terminals 
L and R are thermostats. $k_T$, $k_P$ and $k_D$ are the hopping constants coupling the ring
to the thermostats, the probes and the QD. The energy levels of the dot are $V$, $V+U$ and
$V+2U$.}\label{Fig5}
\end{center}
\end{figure}
Let us write $t_{\rm QD}=\sqrt{T_{\rm QD}} e^{i \alpha_{\rm QD}}$
for the transmission amplitude of the QD. At fixed energy $E$, the total transmission probability 
$t_{LR}$ from the reservoirs L to R depends on the gate voltage $V$ and is a periodic function 
of the AB-flux $\phi$. Expanding this function as a Fourier series one gets:
\begin{equation}\label{Expansion Fourier}
t_{LR}(V,\phi) = A(V) + B(V) \cos(\phi + \beta(V)) + \cdots.
\end{equation}
It is well known that in the absence of dissipation (\ie for a closed interferometer in the terminology of 
\cite{Aharony}) the Onsager-Casimir reciprocity relations \cite{BILP} imply that the phase $\beta(V)$
can only take the values $0$ and $\pi$. Hence, as the gate voltage $V$ varies, the phase $\beta(V)$
makes abrupt jumps between these two values. However, 
dissipation can change this picture. By adding purely absorbing reservoirs (\ie allowing only outgoing 
currents) along the branches of the ring Aharony et al.\;\cite{Aharony} found criteria as to when the 
"experimental" phase $\beta(V)$, which depends on the details of the opening (\ie the coupling the 
absorbing reservoirs), is a good approximation of the "intrinsic" phase $\alpha_{\rm QD}(V)$ of the 
QD. Here we present some numerical results showing that one may capture the main properties
of $\alpha_{\rm QD}$ without introducing any charge dissipation in the absorbing reservoirs, \ie that 
$\beta$ behaves essentially as $\alpha_{\rm QD}$ even if one replaces the absorbing reservoirs
of  \cite{Aharony} by voltage probes, which we recall allow only heat dissipation.
\begin{figure}[htbp]
\begin{center}
\includegraphics[width=12cm]{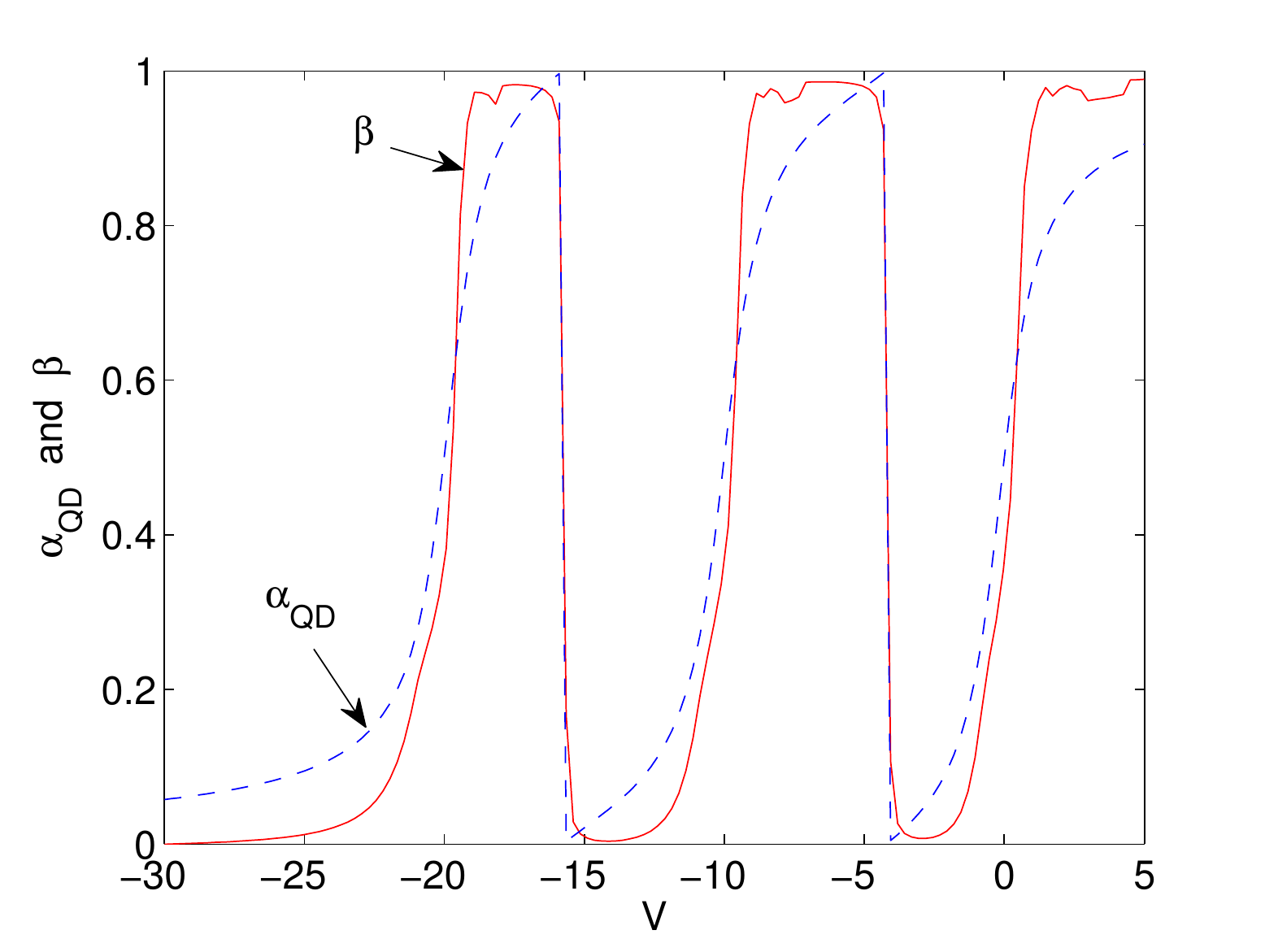}
\caption{The "intrinsic" phase $\alpha_{\rm QD}$ (blue dashed line) and the "experimental" phase 
$\beta$ (red solid line) as a function of the gate voltage $V$ applied on the QD. This corresponds 
to the setup of Figure~\ref{Fig5} with a three-level QD, all temperatures equal to zero, the chemical 
potential of the thermostats are $\mu_{L} = 0$ and $\mu_{R}=0.2$. The couplings to the 
thermostats, probes and QD are $k_{T} = 0.5$, $k_{P} = 0.5 \ k_{T}$, and
 $k_{D} = 0.01 \ k_{T}$, respectively (the parameter $k_{P}$ plays a role 
similar to $J_x$ in \cite{Aharony}). All remaining parameters are set as in \cite{Aharony}.}\label{Fig6}
\end{center}
\end{figure}
However, instead of 
considering the expansion \eref{Expansion Fourier}, we tried to be closer to actual experimental 
measurements by extracting the "experimental" phase $\beta$ from the Fourier expansion of the 
steady electric current between the two thermostats:
$$
I_{L}=-I_{R}=\hat I_0(V)+\hat I_1(V)\cos(\phi+\beta(V))+\cdots.
$$
The results are shown on Figure~\ref{Fig6}. One sees that 
the curve $\beta(V)$ follows relatively closely $\alpha_{\rm QD}(V)$, and in particular reproduces 
accurately the successive jumps of $\alpha_{\rm QD}(V)$ from 1 to 0 (the values have been 
normalized, thus 1 corresponds to $\pi$ in the paper \cite{Aharony}).

\end{document}